\documentclass[journal,twoside,web]{ieeecolor}
\usepackage{generic}
\usepackage[american]{babel}
\usepackage[babel=true]{csquotes}
\usepackage[hidelinks]{hyperref}
\usepackage{amssymb,amsmath,amsfonts,latexsym,mathtools}
\usepackage{xcolor}
\usepackage{array}
\usepackage{textcomp}
\usepackage{stfloats}
\usepackage{url}
\usepackage{graphicx}
\usepackage{booktabs}
\DeclareUnicodeCharacter{2212}{\textminus}

\hypersetup{
  pdftitle={Competitive One-Step-Ahead Control of Friedkin-Johnsen Networks:
  Potential Games, Stability, and the Price of Competition},
  pdfauthor={Gabriel Gentil and Amit Bhaya}
}

\newtheorem{theorem}{Theorem}
\newtheorem{lemma}{Lemma}
\newtheorem{remark}{Remark}
\newtheorem{proposition}{Proposition}

\newtheorem{corollary}{Corollary}
\newtheorem{definition}{Definition}

\makeatletter
\def\tagform@#1{\maketag@@@{\ignorespaces#1\unskip\@@italiccorr}}
\let\orgtheequation\theequation
\def\theequation{(\orgtheequation)}
\makeatother

\def\BibTeX{{\rm B\kern-.05em{\sc i\kern-.025em b}\kern-.08em
    T\kern-.1667em\lower.7ex\hbox{E}\kern-.125emX}}

\begin{document}
	    \title{Competitive One-Step-Ahead Control of Friedkin--Johnsen Networks: Potential Games, Stability, and the Price of Competition}
	\author{Gabriel Gentil, Amit Bhaya%
        \thanks{This work has been submitted to the IEEE for possible publication. Copyright may be transferred without notice, after which this version may no longer be accessible. This work was  partially supported  by CAPES, Finance Code 001 and by the Conselho Nacional de Desenvolvimento Científico e Tecnológico (CNPq), Brazil, under grant BPP/PQ-Sr  313335/2022-2}
        \thanks{The authors are with the Dept.~of Electrical Engineering, Federal University of Rio de Janeiro, PEE/COPPE/UFRJ, PO Box 68504, Rio de Janeiro, RJ 21945-970, Brazil. E-mail: amit@nacad.ufrj.br (Corresponding author).}
	}

\maketitle
\thispagestyle{plain}
\begin{abstract}
This paper studies competitive one-step-ahead control of Friedkin--Johnsen networks with overlapping player influence. The one-step interaction is an exact potential game with a unique Nash equilibrium obtained from a symmetric positive-definite system. Sequential best-response sweeps converge for every frozen network state, parallel sweeps obey an exact Jacobi condition (and always converge with two players), and one-sweep implementations require an
augmented state--action stability test. For marginal networks, a signed left--right damping condition is sufficient for exact-equilibrium stability and becomes a sharp first-order instability test when its sign is reversed. A resolvent identity clarifies the feedback geometry, while a control-aware centrality identifies the goal conflicts that matter most. We characterize attainable equilibria under unconstrained, convex, and sparse goal restrictions and give a closed form for the same-state welfare loss caused by competition. Numerical examples verify the stability thresholds, geometry, and welfare predictions.
\end{abstract}


\section{Introduction}
\label{sec:introduction}

The Friedkin--Johnsen (FJ) model
\cite{friedkin1990social} extends the classical French--Harary and
DeGroot frameworks \cite{French1956,Harary1959,DeGroot1974} by allowing
agents to retain attachment to their initial opinions, thereby
supporting persistent disagreement; see
\cite{Tempo1,Tempo2,Hassani2022,AndersonYe2019} for surveys.
Opinion control has been studied through controllability, pinning, and
influence allocation \cite{Liu2014OpinionControl,Masuda2015OpinionControl}.
Competitive formulations include stubborn groups, Stackelberg games,
and coupled action--opinion dynamics
\cite{GlassGlass2021,Varma_stackelberg,li2023stackelberg,Niazi2021,
Cao_coevolutionary,ye2023coevolution,ShrinateTripathy2025}.
Recent TCNS work studies graph-based targeting, attraction competition,
and dynamic influence maximization
\cite{BiniFrascaRavazziDabbene2022,AoJia2023,
BastopcuEtesamiBasar2025}.

Game-theoretic FJ models commonly let network agents choose their own
opinions and compare equilibrium with a social optimum
\cite{BindelKleinbergOren2015,FotakisKandirosKontonisSkoulakis2023}.
Here the players are instead external decision-makers. Under
one-step-ahead optimal control (OSAOC) \cite{KB2022}, each balances a
predicted tracking error against a quadratic penalty. Previous
competitive OSAOC uses disjoint influence domains
\cite{GG_AB_OSAOC}; related communication-structured games appear in
\cite{JiangMazalovGaoWang23}. With overlapping domains, every scalar
action enters every player's tracking cost, coupling all best
responses.

This coupling raises three questions that are not answered by the
disjoint-domain formulation. First, does the overlapping one-step game
retain a tractable equilibrium structure? Second, do exact,
simultaneous, and sequential implementations have the same convergence
and closed-loop stability properties? Third, how should one distinguish
an equilibrium error caused by structural limitations of the influence
matrix from a welfare loss caused by noncooperative play? The present
paper addresses these questions within a unified potential-game and
closed-loop stability framework. The game is deliberately myopic: each
player minimizes a one-step prediction cost at the current state. It is
therefore distinct from a Markov-perfect dynamic game, in which players
optimize intertemporal objectives while anticipating future strategic
responses.

\subsection{Contributions of the Present Work}
\label{subsec:novelty}

The principal contributions are as follows.
\begin{enumerate}
\item \textbf{Potential-game formulation with overlapping influence.}
The competitive one-step game remains an exact potential game under
overlapping influence and conflicting goals; its unique Nash action is
the solution of a symmetric positive-definite system.

\item \textbf{Separation of equilibrium and implementation.}
Exact, parallel, and sequential protocols have common fixed-point
equations but different convergence properties. We give the exact
frozen-state criteria, prove unconditional parallel convergence for two
players, and give the augmented stability test for one-sweep closed-loop
implementations.

\item \textbf{Closed-loop stability of exact-equilibrium feedback.}
For simple peripheral eigenvalues, a signed left--right damping
condition guarantees Schur stability for sufficiently large penalties.
Its reversed sign gives a sharp first-order large-penalty instability
test. The result covers primitive marginal networks and explains why
unsigned alignment is insufficient for periodic modes.

\item \textbf{Equilibrium geometry and competitive welfare loss.}
We characterize unconstrained and support- or amplitude-constrained
equilibrium families, derive a resolvent representation of the feedback,
introduce a control-aware conflict sensitivity, and separate geometric
target error from the same-state welfare loss of noncooperative play.
\end{enumerate}

\section{Preliminaries}
\label{sec:preliminaries}

\subsection{Notation and Conventions}
\label{subsec:notation}

The sets of real and complex numbers are denoted by $\mathbb{R}$ and
$\mathbb{C}$. The identity and all-ones vectors are denoted by $I$
and $\mathbf{1}$, with dimensions inferred from context; $e_i^{(q)}$
is the $i$th standard basis vector of $\mathbb R^q$ (the superscript is
omitted when unambiguous). For a matrix $X$, $X^\top$ and $X^*$
denote transpose and conjugate transpose. For real symmetric matrices,
$X\succ0$ and $X\succeq0$ denote positive definiteness and
semidefiniteness, whereas $X\geq0$ is entrywise. The symbols
$\|\cdot\|_2$, $\sigma(X)$, and $\rho(X)$ denote the Euclidean or
induced matrix norm, spectrum, and spectral radius. We use
$\operatorname{rank}(X)$, $\ker(X)$, $\operatorname{im}(X)$, and
$\operatorname{diag}(X)$ in their standard senses.

\subsection{Controlled Friedkin--Johnsen Model}
\label{subsec:controlled-fj-model}

The network contains $n$ agents and $p$ players indexed by
$\mathcal{M}:=\{1,\ldots,p\}$. Let $W\in\mathbb{R}^{n\times n}$ be
row-stochastic and let
$\Theta=\operatorname{diag}(\theta_1,\ldots,\theta_n)$, with
$\theta_i\in[0,1]$. Define $A:=(I-\Theta)W$ and
$d:=\Theta x(0)$. With
$B=[\,b_1\ \cdots\ b_p\,]\in\mathbb{R}^{n\times p}$ and
$u(k)\in\mathbb{R}^p$, the controlled dynamics are
\begin{equation}
x(k+1)=Ax(k)+d+Bu(k).
\label{eq:fj_controlled}
\end{equation}
The matrix $A$ is row-substochastic and therefore satisfies
$\rho(A)\leq1$ \cite{Parsegov2017}. The vector $b_m$ is player $m$'s
influence direction. When $b_m$ only identifies targeted nodes and their exposure
strengths, $b_m\geq0$ is natural: the signed direction of intervention
is carried by $u_m\in\mathbb R$, while $g_m$ records the desired signed
opinions. We retain a general real $B$ to allow mixed-sign influence
directions representing antagonistic or contrarian responses.

The boundary case $\rho(A)=1$ has a useful graph interpretation. A set
$\mathcal C\subseteq\{1,\ldots,n\}$ is \emph{closed under $W$} if
$w_{ij}=0$ for every $i\in\mathcal C$ and $j\notin\mathcal C$.

\begin{lemma}[Marginal FJ networks]
\label{lem:marginal-fj}
For $A=(I-\Theta)W$, one has $\rho(A)=1$ if and only if there exists a
nonempty set $\mathcal C$ that is closed under $W$ and contains no
stubborn agent, that is, $\theta_i=0$ for all $i\in\mathcal C$. In
particular, if $A$ is irreducible, then
\[
 \rho(A)=1 \quad\Longleftrightarrow\quad \Theta=0,
\]
in which case $A=W$, $A\mathbf1=\mathbf1$, and $d=0$.
\end{lemma}

\begin{proof}
If such a set $\mathcal C$ exists, then the principal block
$A_{\mathcal C\mathcal C}=W_{\mathcal C\mathcal C}$ is stochastic and
closed, so $1\in\sigma(A)$. Conversely, suppose $Ay=y$ for a nonzero
$y\geq0$, which exists by Perron--Frobenius when $\rho(A)=1$. Let
$M=\max_i y_i>0$ and $\mathcal C=\{i:y_i=M\}$. For $i\in\mathcal C$,
\[
 M=(1-\theta_i)\sum_jw_{ij}y_j\leq(1-\theta_i)M\leq M.
\]
Both inequalities are equalities; hence $\theta_i=0$ and $w_{ij}>0$
implies $j\in\mathcal C$. Thus $\mathcal C$ is nonempty, closed, and
non-stubborn. If $A$ is irreducible, no proper nonempty closed set
exists, so $\mathcal C$ is the full node set and $\Theta=0$. The final
claims follow from the definitions of $A$ and $d$.
\end{proof}

\subsection{The One-Step Game}
\label{subsec:static_game}

At time $k$, define the free response
\begin{equation}
z(k):=Ax(k)+d.
\label{eq:free_response}
\end{equation}
Player $m$ has a goal $g_m\in\mathbb{R}^n$, chooses
$u_m\in\mathbb{R}$, and incurs
\begin{equation}
\begin{aligned}
J_m(u;x(k))
&:=\|z(k)+Bu-g_m\|_2^2\\
&\quad+\gamma_m u_m^2,
\qquad \gamma_m>0.
\end{aligned}
\label{eq:osaoc_cost}
\end{equation}

\begin{definition}[The one-step game]
\label{def:static_game}
The static game parametrized by $x(k)$ is
\[
\mathcal{G}_{\mathrm{OS}}(x(k))
:=
\bigl(\mathcal{M},\{\mathbb{R}\}_{m\in\mathcal{M}},
\{J_m(\cdot;x(k))\}_{m\in\mathcal{M}}\bigr).
\]
\end{definition}

\section{The OSAOC as an Exact Potential Game}
\label{sec:potential_game}

An exact potential has the same change as each player's cost under
every unilateral deviation \cite{monderer1996potential}. The present
quadratic game admits such a potential in closed form.

For compactness, define
\begin{equation}
\begin{aligned}
\Gamma&:=\operatorname{diag}(\gamma_1,\ldots,\gamma_p)\succ0,\\
G&:=B^\top B\succeq0,\\
S&:=G+\Gamma\succ0,
\end{aligned}
\label{eq:game-matrices}
\end{equation}
and
\begin{equation}
v_m:=b_m^\top g_m,
\qquad
v:=[\,v_1\ \cdots\ v_p\,]^\top.
\label{eq:aggregated-individual-goals}
\end{equation}

\begin{lemma}[The OSAOC game is an exact potential game]
	\label{lem:potential_game}
	For every fixed $x(k)$, the game $\mathcal{G}_{\mathrm{OS}}(x(k))$
    is an exact potential game with
	\begin{equation} \label{eq:potential_function}
		\Phi(u) = u^{\top}Su + 2u^{\top}(B^{\top}z(k) - v).
	\end{equation}
	The potential is strictly convex and coercive, and its unique Nash
    equilibrium satisfies $Su^* = v-B^{\top}z(k)$.
\end{lemma}

\begin{proof}
For every player $m$, direct differentiation gives
\[
 \frac{\partial J_m}{\partial u_m}
 =2b_m^{\top}(z(k)+Bu-g_m)+2\gamma_m u_m
 =\frac{\partial\Phi}{\partial u_m}.
\]
Thus $\Phi$ is an exact potential. Moreover,
$\nabla^2\Phi=2S\succ0$, so $\Phi$ is strictly convex and coercive.
Because each $J_m$ is strictly convex in its own scalar action, a
profile is a Nash equilibrium if and only if every partial derivative
$\partial J_m/\partial u_m$ vanishes. By the potential identity, this is
equivalent to $\nabla\Phi=0$, which has the unique solution stated in
the lemma.
\end{proof}

\begin{remark}[Independence of Potential Structure from Goal Conflict]
	\label{rem:goal_conflict}
	Goal conflict changes only the linear term $-2u^{\top}v$ and hence
    displaces the equilibrium; the potential Hessian and uniqueness are
    determined by $B$ and $\Gamma$.
\end{remark}

\section{Equilibrium Computation and Protocol Equivalence}
\label{sec:protocols}

For a fixed network state $x(k)$, define
\begin{equation}
    r(k)
    :=
    v-B^{\top}z(k)
    =
    v-B^{\top}\bigl(Ax(k)+\Theta x(0)\bigr).
    \label{eq:rhs-equilibrium-system}
\end{equation}
By Lemma~\ref{lem:potential_game}, the unique Nash equilibrium of the one-step game
$\mathcal{G}_{\mathrm{OS}}(x(k))$ is the solution of
\begin{equation}
    S u^{\mathrm{NE}}(k)=r(k),
    \qquad
    S=B^{\top}B+\Gamma \succ 0.
    \label{eq:one-step-NE-system}
\end{equation}

We compare exact solution with one parallel or sequential
best-response sweep per network time step.

For each player $m$, the individual best response to a given profile
$u_{-m}$ is
\begin{equation}
    \operatorname{BR}_{m}(u_{-m};x(k))
    =
    \frac{1}{S_{mm}}
    \left(
        r_m(k)-\sum_{j\neq m}S_{mj}u_j
    \right).
    \label{eq:individual-best-response}
\end{equation}
Parallel and sequential updates are, respectively, the Jacobi and
Gauss--Seidel iterations for \eqref{eq:one-step-NE-system}.

Let
\begin{equation}
    S=D+L+L^{\top},
    \label{eq:S-decomposition}
\end{equation}
where $D=\operatorname{diag}(S)$ and $L$ is strictly lower triangular.

\begin{enumerate}
    \item \textbf{Exact Nash-equilibrium feedback (EQ):}
    \begin{equation}
        S u(k)=r(k).
        \label{eq:exact-equilibrium-update}
    \end{equation}
    Thus $M_{\mathrm{EQ}}=S$ and $N_{\mathrm{EQ}}=0$; this represents
    a centralized solve or a converged inner iteration.

    \item \textbf{Parallel best-response protocol (PBR/Jacobi):}
    \begin{equation}
        D u(k)
        =
        -\bigl(L+L^{\top}\bigr)u(k-1)+r(k).
        \label{eq:parallel-BR-update}
    \end{equation}
    Hence $M_{\mathrm{PBR}}=D$ and
    $N_{\mathrm{PBR}}=-(L+L^\top)$.

    \item \textbf{Sequential best-response protocol (SBR/Gauss--Seidel):}
    Players update in the order $1,\ldots,p$, using new actions for
    preceding players:
    \begin{equation}
        (D+L)u(k)
        =
        -L^{\top}u(k-1)+r(k).
        \label{eq:sequential-BR-update}
    \end{equation}
    Thus $M_{\mathrm{SBR}}=D+L$ and
    $N_{\mathrm{SBR}}=-L^\top$.
\end{enumerate}

All three protocols can be written in the common matrix-splitting
form
\begin{equation}\label{eq:generic-protocol-update}
    M_{\nu} u(k)
    =
    N_{\nu} u(k-1)+r(k),
    \qquad
    S=M_{\nu}-N_{\nu},
\end{equation}
where $\nu\in\{\mathrm{EQ},\mathrm{PBR},\mathrm{SBR}\}$. Their
transients differ, but their fixed points agree whenever convergence
occurs.

\begin{theorem}[Common fixed-point equations]
\label{thm:common-fixed-point}
For every protocol
$\nu\in\{\mathrm{EQ},\mathrm{PBR},\mathrm{SBR}\}$, the set of fixed
points of the corresponding closed-loop update is characterized by
the common equations
\begin{subequations}
\label{eq:common-fixed-point}
\begin{align}
    x^{*}
    &=
    Ax^{*}+\Theta x(0)+Bu^{*},
    \label{eq:common-fixed-point-state}
    \\
    Su^{*}
    &=
    v-B^{\top}\bigl(Ax^{*}+\Theta x(0)\bigr).
    \label{eq:common-fixed-point-control}
\end{align}
\end{subequations}
Consequently, the choice of protocol does not change the closed-loop
fixed-point equations. It only changes the transient trajectory and
the conditions under which that fixed point is reached.
\end{theorem}

\begin{proof}
At a fixed point, the state recursion gives
\eqref{eq:common-fixed-point-state}, while
\eqref{eq:generic-protocol-update} and
$M_\nu-N_\nu=S$ give \eqref{eq:common-fixed-point-control}.
Conversely, these two equations reproduce the fixed-point form of the
state and protocol updates for every $\nu$.
\end{proof}

\begin{remark}[Fixed-point equivalence does not imply convergence]
\label{rem:fixed-point-versus-convergence}
Common fixed points do not imply common convergence: frozen-state PBR
and SBR are governed by their splitting matrices, whereas one-sweep
closed-loop convergence is governed by the augmented dynamics in
Section~\ref{sec:coupling_caveat}.
\end{remark}

\section{Dynamics and Closed-Loop Stability}
\label{sec:stability}

Analyzing the stability of the competitive OSAOC game requires distinguishing between the convergence of the players' control updates at a fixed point in time and the true closed-loop stability of the evolving state-action system. 

\subsection{Convergence of Best-Response Sweeps at a Frozen State}
\label{subsec:frozen-state-convergence}

We first analyze the computational convergence of the
best-response protocols while holding the network state fixed.
This analysis is distinct from the closed-loop problem in which the
state evolves after every best-response sweep.

Fix a state $\bar{x}\in\mathbb{R}^{n}$ and define
\begin{equation}
    \bar{z}
    :=
    A\bar{x}+\Theta x(0),
    \qquad
    \bar{r}
    :=
    v-B^{\top}\bar{z}.
    \label{eq:frozen-state-rhs}
\end{equation}
The Nash equilibrium associated with the frozen state is the unique
solution of
\begin{equation}
    S u^{\mathrm{NE}}=\bar{r},
    \qquad
    S=B^{\top}B+\Gamma\succ0.
    \label{eq:frozen-state-NE-system}
\end{equation}

To avoid confusing the best-response iteration with the network
time index $k$, let $\ell=0,1,2,\ldots$ denote the inner-sweep
index. For a matrix splitting
\begin{equation}
    S=M_{\nu}-N_{\nu},
    \label{eq:frozen-matrix-splitting}
\end{equation}
the corresponding stationary iteration is
\begin{equation}
    M_{\nu} u^{(\ell+1)}
    =
    N_{\nu} u^{(\ell)}+\bar{r}.
    \label{eq:frozen-stationary-iteration}
\end{equation}
Subtracting
\[
    M_{\nu} u^{\mathrm{NE}}
    =
    N_{\nu} u^{\mathrm{NE}}+\bar{r}
\]
from~\eqref{eq:frozen-stationary-iteration} gives the error
recursion
\begin{equation}
\begin{aligned}
    e^{(\ell+1)}
    &=
    T_{\nu} e^{(\ell)},
    &T_{\nu}&:=M_{\nu}^{-1}N_{\nu},\\
    e^{(\ell)}&:=u^{(\ell)}-u^{\mathrm{NE}}.
\end{aligned}
    \label{eq:frozen-error-recursion}
\end{equation}
Therefore, the iteration converges to $u^{\mathrm{NE}}$ from every
initialization if and only if
\begin{equation}
    \rho(T_{\nu})<1.
    \label{eq:frozen-convergence-condition}
\end{equation}

Using the decomposition
\[
    S=D+L+L^{\top},
    \qquad
    D=\operatorname{diag}(S),
\]
the three protocols have the splittings
\begin{align}
    M_{\mathrm{EQ}}
    &=
    S,
    &
    N_{\mathrm{EQ}}
    &=
    0,
    \label{eq:EQ-splitting}
    \\
    M_{\mathrm{PBR}}
    &=
    D,
    &
    N_{\mathrm{PBR}}
    &=
    -\bigl(L+L^{\top}\bigr)
    =
    D-S,
    \label{eq:PBR-splitting}
    \\
    M_{\mathrm{SBR}}
    &=
    D+L,
    &
    N_{\mathrm{SBR}}
    &=
    -L^{\top}
    =
    M_{\mathrm{SBR}}-S.
    \label{eq:SBR-splitting}
\end{align}

\begin{proposition}[Frozen-state convergence]
\label{prop:frozen-state-convergence}
Let
\[
    S=B^{\top}B+\Gamma\succ0,
    \qquad
    D=\operatorname{diag}(S).
\]
Then the following statements hold.

\begin{enumerate}
    \item \textbf{Exact equilibrium feedback:}
    The EQ protocol obtains the Nash equilibrium in one solve:
    \[
        u^{\mathrm{NE}}=S^{-1}\bar{r}.
    \]
    Its iteration matrix is
    \[
        T_{\mathrm{EQ}}=0.
    \]

    \item \textbf{Sequential best response:}
    The SBR/Gauss--Seidel iteration converges to
    $u^{\mathrm{NE}}$ from every initialization:
    \begin{equation}
        \rho\bigl(T_{\mathrm{SBR}}\bigr)
        =
        \rho\left(
            -(D+L)^{-1}L^{\top}
        \right)
        <1.
        \label{eq:SBR-unconditional-convergence}
    \end{equation}

    \item \textbf{Parallel best response:}
    The PBR/Jacobi iteration converges to $u^{\mathrm{NE}}$ from
    every initialization if and only if
    \begin{equation}
        2D-S\succ0.
        \label{eq:PBR-exact-condition}
    \end{equation}
    Since
    \[
        S=G+\Gamma,
        \qquad
        D=\operatorname{diag}(G)+\Gamma,
    \]
    condition~\eqref{eq:PBR-exact-condition} is equivalently
    \begin{equation}
        2\operatorname{diag}(G)+\Gamma-G\succ0.
        \label{eq:PBR-Gram-condition}
    \end{equation}
\end{enumerate}
\end{proposition}

\begin{proof}
The EQ claim is immediate. For SBR,
$M_{\mathrm{SBR}}+M_{\mathrm{SBR}}^\top-S=D\succ0$; hence the
Householder--John theorem for positive-definite splittings
\cite[p.~137]{Salgado_Wise2023} gives
$\rho(T_{\mathrm{SBR}})<1$. For PBR,
\[
T_{\mathrm{PBR}}=I-D^{-1}S
 \sim I-D^{-1/2}SD^{-1/2}.
\]
Since $D^{-1/2}SD^{-1/2}\succ0$, all eigenvalues of
$T_{\mathrm{PBR}}$ lie in $(-1,1)$ exactly when
$D^{-1/2}SD^{-1/2}\prec2I$, equivalently $2D-S\succ0$.
Substitution of $S=G+\Gamma$ gives
\eqref{eq:PBR-Gram-condition}.
\end{proof}

These are frozen-state results; one-sweep closed-loop convergence
requires the augmented analysis in Section~\ref{sec:coupling_caveat}.

\begin{proposition}[Two-player parallel sweeps]
\label{prop:two-player-pbr}
If $p=2$, the frozen-state PBR iteration converges for every
$B$ and every $\Gamma\succ0$.
\end{proposition}

\begin{proof}
Write $S=\left[\begin{smallmatrix}s_{11}&s_{12}\\s_{12}&s_{22}
\end{smallmatrix}\right]\succ0$. Then
$2D-S=\left[\begin{smallmatrix}s_{11}&-s_{12}\\-s_{12}&s_{22}
\end{smallmatrix}\right]$ has the same positive leading principal
minors as $S$. Thus $2D-S\succ0$, and
Proposition~\ref{prop:frozen-state-convergence} applies.
\end{proof}

For general $p$, a convenient sufficient (though not necessary)
condition follows from strict diagonal dominance:
\begin{equation}
 \gamma_m+\|b_m\|_2^2>
 \sum_{j\ne m}|b_m^\top b_j|,
 \qquad m=1,\ldots,p.
\label{eq:pbr-diagonal-dominance}
\end{equation}
A single conservative bound is
\begin{equation}
\min_m\gamma_m>
\rho\!\left(G-2\operatorname{diag}(G)\right),
\label{eq:pbr-scalar-bound}
\end{equation}
because it implies
$\Gamma-(G-2\operatorname{diag}(G))\succ0$.

\subsection{Exact Nash-Equilibrium Closed-Loop Stability}
\label{subsec:EQ-stability}

Under the exact Nash-equilibrium feedback protocol introduced in
Section~\ref{sec:protocols}, the control vector is
\begin{equation}
    u_{\mathrm{EQ}}(k)
    =
    S^{-1}
    \left[
        v-B^{\top}\bigl(Ax(k)+\Theta x(0)\bigr)
    \right].
    \label{eq:EQ-control-law}
\end{equation}
Substitution into the state equation gives
\begin{equation}
    x(k+1)
    =
    F_{\mathrm{EQ}}(\Gamma)x(k)
    +
    c_{\mathrm{EQ}}(\Gamma),
    \label{eq:EQ-closed-loop}
\end{equation}
where
\begin{align}
    F_{\mathrm{EQ}}(\Gamma)
    &=
    \bigl(I-P_{\mathrm{EQ}}(\Gamma)\bigr)A,
    \label{eq:EQ-transition-matrix}
    \\
    P_{\mathrm{EQ}}(\Gamma)
    &=
    B\bigl(B^{\top}B+\Gamma\bigr)^{-1}B^{\top},
    \label{eq:EQ-influence-operator}
    \\
    c_{\mathrm{EQ}}(\Gamma)
    &=
    \bigl(I-P_{\mathrm{EQ}}(\Gamma)\bigr)\Theta x(0)
    +
    B\bigl(B^{\top}B+\Gamma\bigr)^{-1}v.
    \label{eq:EQ-affine-term}
\end{align}

The next identity exposes the exact feedback geometry and will also be
used in the reachability and welfare analyses.

\begin{lemma}[Resolvent form of exact-equilibrium feedback]
\label{lem:resolvent}
Let $C(\Gamma):=B\Gamma^{-1}B^\top$. Then
\begin{equation}
 I-P_{\mathrm{EQ}}(\Gamma)=(I+C(\Gamma))^{-1},
 \qquad
 F_{\mathrm{EQ}}(\Gamma)=(I+C(\Gamma))^{-1}A.
\label{eq:resolvent-identity}
\end{equation}
Consequently $0\preceq P_{\mathrm{EQ}}(\Gamma)\prec I$, and
$P_{\mathrm{EQ}}(\Gamma)$ is an orthogonal projector only when $B=0$.
If $\Gamma=\varepsilon\bar\Gamma$ and $\varepsilon\downarrow0$, then
$P_{\mathrm{EQ}}(\Gamma)\to\Pi_B$, the orthogonal projector onto
$\operatorname{im}(B)$, and
\begin{equation}
 F_{\mathrm{EQ}}(\varepsilon\bar\Gamma)
 \longrightarrow (I-\Pi_B)A.
\label{eq:small-penalty-limit}
\end{equation}
\end{lemma}

\begin{proof}
The first identity is the Woodbury formula. With
$\widetilde B=B\Gamma^{-1/2}$, the nonzero eigenvalues of
$P_{\mathrm{EQ}}=\widetilde B(I+\widetilde B^\top\widetilde B)^{-1}
\widetilde B^\top$ are $\sigma_i^2/(1+\sigma_i^2)\in(0,1)$.
This proves the order and idempotence claims. The weighted singular
value decomposition, or equivalently the same eigenvalue formula after
scaling $\Gamma$ by $\varepsilon$, gives the limit
$P_{\mathrm{EQ}}\to\Pi_B$.
\end{proof}

Thus the closed-loop eigenproblem can equivalently be written as the
generalized pencil
\begin{equation}
 Aw=\lambda(I+C(\Gamma))w,
\label{eq:closed-loop-pencil}
\end{equation}
which is also a natural starting point when peripheral eigenvalues are
repeated.

We analyze the large-penalty regime by fixing the relative penalty
weights and scaling their common magnitude. Let
\begin{equation}
    \Gamma(t)=t\bar{\Gamma},
    \qquad
    \bar{\Gamma}
    =
    \operatorname{diag}
    (\bar{\gamma}_1,\ldots,\bar{\gamma}_p)
    \succ 0,
    \qquad
    t>0.
    \label{eq:scaled-regularization}
\end{equation}
The matrix $\bar{\Gamma}$ specifies the relative penalties of the
players, whereas $t$ determines their common scale.

For a matrix $A$ satisfying $\rho(A)\leq 1$, define its peripheral
spectrum by
\begin{equation}
    \sigma_{\mathrm{per}}(A)
    :=
    \left\{
        \mu\in\sigma(A):|\mu|=1
    \right\}.
    \label{eq:peripheral-spectrum}
\end{equation}

\begin{definition}[Primitive and periodic nonnegative matrices]
\label{def:primitive-periodic}
For irreducible $A\geq0$, its period is
\[
h(A):=\gcd\{k\geq1:(A^k)_{ii}>0\},
\]
which is independent of $i$. The matrix is \emph{primitive} if
$h(A)=1$ and \emph{periodic} (or imprimitive) if $h(A)>1$
\cite{HornJohnson2012}. If $\rho(A)=1$ and $h(A)=h$, its peripheral
eigenvalues are the $h$th roots of unity (under the irreducibility
hypothesis above).
\end{definition}

\begin{theorem}[Large-penalty stability of exact equilibrium feedback]
\label{thm:EQ-large-penalty-stability}
Let $A$ be row-substochastic, so that $\rho(A)\leq 1$, and let
\[
    F_{\mathrm{EQ}}(t)
    :=
    F_{\mathrm{EQ}}\bigl(\Gamma(t)\bigr)
    =
    \bigl(I-P_{\mathrm{EQ}}(t)\bigr)A,
\]
where
\[
    P_{\mathrm{EQ}}(t)
    =
    B\bigl(B^{\top}B+t\bar{\Gamma}\bigr)^{-1}B^{\top}.
\]
Then the following statements hold.

\begin{enumerate}
    \item If $\rho(A)<1$, there exists $t_0>0$ such that
    \[
        \rho\bigl(F_{\mathrm{EQ}}(t)\bigr)<1
        \qquad
        \text{for every }t\geq t_0.
    \]

    \item Suppose that $\rho(A)=1$ and that every eigenvalue
    $\mu\in\sigma_{\mathrm{per}}(A)$ is simple. For each such
    eigenvalue, let $w_\mu,y_\mu\in\mathbb{C}^{n}$ satisfy
    \begin{equation}
        Aw_\mu=\mu w_\mu,
        \qquad
        y_\mu^{*}A=\mu y_\mu^{*},
        \qquad
        y_\mu^{*}w_\mu=1,
        \label{eq:peripheral-eigenvectors}
    \end{equation}
    where ${}^{*}$ denotes conjugate transpose. Define
    \begin{equation}
        C
        :=
        B\bar{\Gamma}^{-1}B^{\top}.
        \label{eq:first-order-feedback-matrix}
    \end{equation}
    If the spectral damping condition
    \begin{equation}
        \alpha_\mu
        :=
        \operatorname{Re}
        \left(
            y_\mu^{*}Cw_\mu
        \right)
        >0
        \label{eq:spectral-damping-condition}
    \end{equation}
    holds for every $\mu\in\sigma_{\mathrm{per}}(A)$, then there
    exists $t_0>0$ such that
    \begin{equation}
        \rho\bigl(F_{\mathrm{EQ}}(t)\bigr)<1
        \qquad
        \text{for every }t\geq t_0.
        \label{eq:large-penalty-stability-conclusion}
    \end{equation}
    Conversely, if \eqref{eq:peripheral-eigenvectors} holds and
    $\alpha_\mu<0$ for at least one peripheral mode, then
    \begin{equation}
    \begin{aligned}
      \rho\bigl(F_{\mathrm{EQ}}(t)\bigr)
      &\geq 1+\frac{|\alpha_\mu|}{t}+O(t^{-2})>1\\[-1mm]
      &\hspace{8mm}\text{for all sufficiently large }t.
    \end{aligned}
      \label{eq:large-penalty-instability-converse}
    \end{equation}
    Thus the sign test is sharp to first order. When
    $\alpha_\mu=0$, higher-order terms are decisive.
\end{enumerate}
\end{theorem}

\begin{proof}
See Appendix~\ref{app:large-penalty-proofs}.
\end{proof}

\begin{corollary}[Primitive marginal networks]
\label{cor:primitive-EQ-stability}
Suppose that $A$ is primitive and row-substochastic with
$\rho(A)=1$. Let $B\geq 0$ and assume that $B$ has at least one
nonzero column. Then, for every fixed $\bar{\Gamma}\succ0$, there
exists $t_0>0$ such that
\[
    \rho\bigl(F_{\mathrm{EQ}}(t)\bigr)<1
    \qquad
    \text{for every }t\geq t_0.
\]
\end{corollary}

\begin{proof}
The Perron mode has $w_1=\mathbf1$ and a normalized left eigenvector
$\pi>0$ with $\pi^\top\mathbf1=1$. Indeed,
Lemma~\ref{lem:marginal-fj} and irreducibility imply $\Theta=0$ and
$A\mathbf1=\mathbf1$. Moreover,
\begin{align}
    \pi^{\top}
    B\bar{\Gamma}^{-1}B^{\top}
    \mathbf{1}
    &=
    \sum_{m=1}^{p}
    \frac{
        (\pi^{\top}b_m)
        (b_m^{\top}\mathbf{1})
    }{\bar{\gamma}_m}.
    \label{eq:primitive-damping-sum}
\end{align}
It is positive because $B\geq0$ has a nonzero column. The result now
follows from Theorem~\ref{thm:EQ-large-penalty-stability}.
\end{proof}

\begin{remark}[Signed modal damping]
Condition~\eqref{eq:spectral-damping-condition} uses both left and
right eigenvectors; $B^\top w_\mu\ne0$ alone does not determine
whether a periodic mode moves inward. Also, when $\rho(A)=1$,
$F_{\mathrm{EQ}}(t)\to A$: the theorem is existential and the
stability margin generally vanishes as $t\to\infty$, so any finite
penalty must be checked directly.
\end{remark}
\begin{corollary}[Independently varying player penalties]
\label{cor:independent-penalty-stability}
Suppose that the hypotheses concerning the peripheral eigenvalues
in Theorem~\ref{thm:EQ-large-penalty-stability} hold. Let
\[
    \mathcal{I}
    :=
    \left\{
        m\in\{1,\ldots,p\}:b_m\neq0
    \right\}
\]
be the set of players with nonzero influence directions. Assume that
\begin{equation}
    \alpha_{\mu m}
    :=
    \operatorname{Re}
    \left[
        (y_\mu^{*}b_m)
        (b_m^{\top}w_\mu)
    \right]
    >0
    \label{eq:columnwise-damping}
\end{equation}
for every $\mu\in\sigma_{\mathrm{per}}(A)$ and every
$m\in\mathcal{I}$. Then there exists $\gamma^{*}>0$ such that
\begin{equation}
    \min_{m\in\mathcal{I}}\gamma_m>\gamma^{*}
    \quad\Longrightarrow\quad
    \rho\bigl(F_{\mathrm{EQ}}(\Gamma)\bigr)<1.
    \label{eq:independent-penalty-conclusion}
\end{equation}
\end{corollary}

\begin{proof}
See Appendix~\ref{app:large-penalty-proofs}.
\end{proof}

\subsection{Iterative Modes and the Coupling Caveat}\label{sec:coupling_caveat}
In resource-constrained settings, the players may perform only one
PBR or SBR sweep per network time step. The state then evolves
simultaneously with the action update, so the two recursions must be
analyzed as one augmented system.

\begin{theorem}[Closed-loop stability of one-sweep protocols]
	\label{thm:augmented_stability}
For $\nu\in\{\mathrm{PBR},\mathrm{SBR}\}$, define
\[
Q_\nu:=BM_\nu^{-1}B^\top,
\qquad
\xi(k):=\begin{bmatrix}x(k)\\u(k-1)\end{bmatrix}.
\]
The one-sweep closed loop is
\begin{equation}
\xi(k+1)=\mathcal A_\nu\xi(k)+c_\nu,
\label{eq:augmented-dynamics}
\end{equation}
where
\begin{equation}
\mathcal A_\nu=
\begin{bmatrix}
(I-Q_\nu)A & BM_\nu^{-1}N_\nu\\
-M_\nu^{-1}B^\top A & M_\nu^{-1}N_\nu
\end{bmatrix}
\label{eq:augmented_matrix}
\end{equation}
and
\begin{equation}
c_\nu=
\begin{bmatrix}
(I-Q_\nu)d+BM_\nu^{-1}v\\
M_\nu^{-1}(v-B^\top d)
\end{bmatrix}.
\label{eq:augmented-affine-term}
\end{equation}
The affine system is globally attracted to a fixed point if and only
if $\rho(\mathcal A_\nu)<1$; under this condition the fixed point is
necessarily unique.
\end{theorem}
\begin{proof}
Solving the splitting update \eqref{eq:generic-protocol-update} gives
\begin{align}
u(k)
&=M_\nu^{-1}N_\nu u(k-1)-M_\nu^{-1}B^\top Ax(k)
\nonumber\\
&\quad+M_\nu^{-1}(v-B^\top d).
\label{eq:one-sweep-control-expanded}
\end{align}
Substitution into \eqref{eq:fj_controlled} yields the upper block of
\eqref{eq:augmented-dynamics}; \eqref{eq:one-sweep-control-expanded}
is its lower block. The resulting affine time-invariant system
converges from every initial condition if and only if its homogeneous
transition matrix is Schur stable, which is equivalent to
$\rho(\mathcal A_\nu)<1$.
\end{proof}

In particular, $(I-Q_\nu)A$ is only a diagonal block of
$\mathcal A_\nu$; one-sweep stability must be tested using
$\rho(\mathcal A_\nu)$, not the diagonal blocks separately. Neither
frozen-state convergence nor one-sweep stability implies the other.
For the reverse separation, take
\[
A=
\begin{bmatrix}
0&1\\
0.3&0.6
\end{bmatrix},
\qquad
B=
\begin{bmatrix}
1&1&1\\
0&0&0
\end{bmatrix},
\qquad
\Gamma=\frac12 I_3 .
\]
This $A$ is a valid row-substochastic FJ matrix. Here
\[
\sigma(T_{\mathrm{PBR}})
=
\left\{-\frac43,\frac23,\frac23\right\},
\qquad
\rho(T_{\mathrm{PBR}})=\frac43>1,
\]
whereas
\[
\sigma(\mathcal A_{\mathrm{PBR}})
\approx
\left\{
0.723109,\,
-0.728221\pm0.151198\,\mathrm{i},\,
\frac23,\frac23
\right\},
\]
so $\rho(\mathcal A_{\mathrm{PBR}})=0.743752<1$. Thus the reverse
separation persists under the natural restriction $B\geq0$.

\section{Equilibrium Geometry and Welfare}
\label{sec:equilibrium-geometry}

The implementation protocols introduced in
Section~\ref{sec:protocols} generally have different transient
dynamics. Nevertheless, Theorem~\ref{thm:common-fixed-point} shows
that their fixed-point equations are identical. Consequently, the
set of closed-loop equilibria is determined by the Nash equilibrium
matrix $S$, rather than by the mode-specific splitting matrices
$M_{\nu}$ and $N_{\nu}$.

Throughout this section, define
\begin{equation}
    d:=\Theta x(0),
    \qquad
    S:=B^{\top}B+\Gamma.
    \label{eq:equilibrium-geometry-operators}
\end{equation}
The operator $P_{\mathrm{EQ}}=BS^{-1}B^\top$ introduced in
\eqref{eq:EQ-influence-operator} is common to the EQ, PBR, and SBR
fixed-point equations.

\subsection{The Common Closed-Loop Equilibrium Map}
\label{subsec:common-equilibrium-map}

At a fixed point, the state and control vectors satisfy
\begin{subequations}
\label{eq:fixed-point-system-repeated}
\begin{align}
    x^{*}
    &=
    Ax^{*}+d+Bu^{*},
    \label{eq:fixed-point-state-repeated}
    \\
    Su^{*}
    &=
    v-B^{\top}(Ax^{*}+d).
    \label{eq:fixed-point-control-repeated}
\end{align}
\end{subequations}
Solving~\eqref{eq:fixed-point-control-repeated} for the equilibrium
control gives
\begin{equation}
    u^{*}
    =
    S^{-1}v
    -
    S^{-1}B^{\top}(Ax^{*}+d).
    \label{eq:equilibrium-control-eliminated}
\end{equation}
Substitution into~\eqref{eq:fixed-point-state-repeated} yields
\begin{align}
    x^{*}
    &=
    Ax^{*}+d
    +
    BS^{-1}v
    -
    BS^{-1}B^{\top}(Ax^{*}+d)
    \nonumber\\
    &=
    (I-P_{\mathrm{EQ}})Ax^{*}
    +
    (I-P_{\mathrm{EQ}})d
    +
    BS^{-1}v.
    \label{eq:equilibrium-state-elimination}
\end{align}

The equilibrium resolvent is
\begin{equation}
    K
    :=
    I-F_{\mathrm{EQ}}
    =
    I-(I-P_{\mathrm{EQ}})A.
    \label{eq:equilibrium-resolvent}
\end{equation}

\begin{proposition}[Protocol-independent equilibrium map]
\label{prop:protocol-independent-equilibrium-map}
Suppose that
\begin{equation}
    K=I-(I-P_{\mathrm{EQ}})A
    \label{eq:K-nonsingular-assumption}
\end{equation}
is nonsingular. Then every implementation protocol
$\nu\in\{\mathrm{EQ},\mathrm{PBR},\mathrm{SBR}\}$ has the same
closed-loop fixed point, given by
\begin{equation}
    x^{*}(v)
    =
    x_{\mathrm{base}}+Hv,
    \label{eq:common-affine-equilibrium-map}
\end{equation}
where
\begin{align}
    x_{\mathrm{base}}
    &:=
    K^{-1}(I-P_{\mathrm{EQ}})d,
    \label{eq:zero-goal-baseline}
    \\
    H
    &:=
    K^{-1}BS^{-1}.
    \label{eq:common-equilibrium-sensitivity}
\end{align}
The corresponding equilibrium control is
\begin{equation}
    u^{*}(v)
    =
    S^{-1}
    \left[
        v-B^{\top}\bigl(Ax^{*}(v)+d\bigr)
    \right].
    \label{eq:common-equilibrium-control}
\end{equation}
\end{proposition}

\begin{proof}
Equation~\eqref{eq:equilibrium-state-elimination} gives
$Kx^*=(I-P_{\mathrm{EQ}})d+BS^{-1}v$. Multiplication by $K^{-1}$ yields
\eqref{eq:common-affine-equilibrium-map}; the control follows from
\eqref{eq:fixed-point-control-repeated}. Only the common fixed-point
equations are used, so the map is protocol independent.
\end{proof}

\begin{remark}[Relationship with closed-loop stability]
\label{rem:equilibrium-map-and-stability}
For EQ, $\rho(F_{\mathrm{EQ}})<1$ implies that
$K=I-F_{\mathrm{EQ}}$ is nonsingular. Likewise,
$\rho(\mathcal A_\nu)<1$ for a one-sweep protocol implies
nonsingularity of $K$, since a vector in $\ker(K)$ would generate an
eigenvector of $\mathcal A_\nu$ associated with one: explicitly, if
$Kx=0$ then
\[
 \xi=\operatorname{col}\bigl(x,-S^{-1}B^\top Ax\bigr)
\]
satisfies $\mathcal A_\nu\xi=\xi$ for either one-sweep splitting.
\end{remark}

\subsection{The Reachable Equilibrium Set}
\label{subsec:reachable-equilibrium-set}

Assume that the aggregated goal vector $v$ may vary over
$\mathbb{R}^{p}$. The set of equilibrium states that can be produced
by varying $v$ is
\begin{align}
    \mathcal{E}
    &:=
    \left\{
        x^{*}(v):v\in\mathbb{R}^{p}
    \right\}
    \nonumber\\
    &=
    \left\{
        x_{\mathrm{base}}+Hv:
        v\in\mathbb{R}^{p}
    \right\}
    \nonumber\\
    &=
    x_{\mathrm{base}}+\operatorname{im}(H).
    \label{eq:common-reachable-equilibrium-set}
\end{align}
Thus, $\mathcal{E}$ is an affine subspace. It is common to every
implementation protocol that converges to the Nash fixed point.

\begin{proposition}[Dimension of the reachable equilibrium set]
\label{prop:reachable-equilibrium-dimension}
Suppose that $K$ is nonsingular. Then
\begin{equation}
    \operatorname{rank}(H)
    =
    \operatorname{rank}(B),
    \label{eq:H-rank-equals-B-rank}
\end{equation}
and therefore
\begin{equation}
    \dim(\mathcal{E})
    =
    \operatorname{rank}(B).
    \label{eq:reachable-set-dimension}
\end{equation}
In particular, $\mathcal{E}$ is a proper affine subspace of
$\mathbb{R}^{n}$ if and only if
\begin{equation}
    \operatorname{rank}(B)<n.
    \label{eq:proper-reachable-subspace}
\end{equation}
\end{proposition}

\begin{proof}
Both factors flanking $B$ in $H=K^{-1}BS^{-1}$ are nonsingular, so
$\operatorname{rank}(H)=\operatorname{rank}(B)$. The dimension claim
follows from $\mathcal E=x_{\mathrm{base}}+\operatorname{im}(H)$.
\end{proof}

\begin{remark}[Equilibrium reachability]
\label{rem:equilibrium-versus-dynamic-reachability}
The set $\mathcal{E}$ describes steady-state equilibria obtained by
varying the strategic goal vector $v$. It is not the finite-time
reachable set associated with the classical controllability pair
$(A,B)$. We therefore use the term \emph{equilibrium reachability}
when a distinction is necessary.
\end{remark}

\subsection{Equilibrium Geometry under Constrained Goals}
\label{subsec:constrained-goal-geometry}

Goal constraints can be represented directly in the reduced variable $v$.
For node $i$, define
\[
\begin{aligned}
    \boldsymbol{g}_i
      &:=\bigl[(g_1)_i\ \cdots\ (g_p)_i\bigr]^\top,\\
    D_i&:=\operatorname{diag}(B_{i1},\ldots,B_{ip}),
\end{aligned}
\]
so that $v=\sum_{i=1}^nD_i\boldsymbol{g}_i$.  To separate agreement from
conflict, write
\[
    \boldsymbol{g}_i=\alpha_i\mathbf{1}+c_i,
    \qquad \mathbf{1}^\top c_i=0,
\]
and let $U_\perp\in\mathbb{R}^{p\times(p-1)}$ have orthonormal columns
spanning $\mathbf{1}^\perp$.
For a selected support $\mathcal C=\{i_1,\ldots,i_r\}$, let
$\Delta_{\mathcal C}\subseteq\mathbb R^{r(p-1)}$ denote the admissible
set of stacked conflict coordinates. Its dimension therefore follows
the cardinality of $\mathcal C$.

\begin{proposition}[Support-restricted goal conflict]
\label{prop:support-restricted-goals}
Suppose that $K$ is nonsingular, the common components $\alpha_i$ are fixed,
and conflict is permitted only on
$\mathcal{C}=\{i_1,\ldots,i_r\}$.  With $c_i=U_\perp\eta_i$, define
\[
    T_{\mathcal{C}}
    :=\bigl[D_{i_1}U_\perp\ \cdots\ D_{i_r}U_\perp\bigr],
    \qquad
    \eta:=\operatorname{col}(\eta_{i_1},\ldots,\eta_{i_r}).
\]
Typical choices are $\Delta_{\mathcal C}=\mathbb R^{r(p-1)}$, a box,
or $\Delta_{\mathcal C}=\{\eta:\|\eta\|_2\leq\varepsilon\}$.
The corresponding equilibrium set is
\begin{equation}
    \mathcal{E}_{\mathcal{C},\Delta_{\mathcal C}}
    =x^*(v_0)+HT_{\mathcal{C}}\Delta_{\mathcal{C}},
    \qquad
    v_0:=B^\top\alpha,
    \label{eq:conflict-restricted-equilibrium-set}
\end{equation}
where $\alpha=\operatorname{col}(\alpha_1,\ldots,\alpha_n)$.
If $\Delta_{\mathcal{C}}=\mathbb{R}^{r(p-1)}$, this set is affine and
\begin{equation}
    \dim(\mathcal{E}_{\mathcal{C},\Delta_{\mathcal C}})
    =\operatorname{rank}(HT_{\mathcal{C}})
    \leq\min\{\operatorname{rank}(B),r(p-1)\}.
    \label{eq:conflict-restricted-dimension}
\end{equation}
If $\Delta_{\mathcal{C}}$ is compact and convex, its equilibrium image is
compact and convex; box and Euclidean-ball constraints produce, respectively,
a possibly degenerate zonotope and ellipsoid.
\end{proposition}

This follows by substituting $v=v_0+T_{\mathcal C}\eta$ into the
affine equilibrium map. Thus fixed support alone preserves affinity,
whereas bounded amplitudes generally do not. If conflict may occur on
any set of at most $r$ nodes, the attainable family is the generally
nonconvex union
\begin{equation}
    \mathcal{E}_{\leq r}
    =\bigcup_{|\mathcal{C}|\leq r}
      \mathcal{E}_{\mathcal{C},\Delta_{\mathcal C}}.
    \label{eq:sparse-conflict-equilibrium-set}
\end{equation}
Projection onto a fixed convex image is a convex least-squares problem,
whereas projection onto this union is generally nonconvex and may be
nonunique. Without zero-sum coupling, unbounded support restrictions
can even recover the full family whenever every player has a freely
variable goal at an influenced selected node.

Selecting $\mathcal{C}$ solely from a centrality measure of $A$ may therefore
be misleading.  The relevant equilibrium sensitivities are
\begin{equation}
\begin{aligned}
    \frac{\partial x^*}{\partial(g_m)_i}&=B_{im}He_m^{(p)},\\
    \max_{\|\eta\|_2\leq\varepsilon}
    \|x^*-x^*(v_0)\|_2
      &=\varepsilon\|HT_{\mathcal{C}}\|_2.
\end{aligned}
    \label{eq:conflict-amplification}
\end{equation}
Hence $\|HD_iU_\perp\|_2$ is a control-aware single-node score and
$\operatorname{rank}(HT_{\mathcal C})$ counts the induced equilibrium
directions. Pure zero-sum conflict leaves the aggregate social goal
unchanged, so its welfare effect is a convex quadratic in $\eta$.

\subsection{Closest Reachable Equilibrium}
\label{subsec:closest-reachable-equilibrium}

\begin{proposition}[Orthogonal projection onto the equilibrium set]
\label{prop:closest-reachable-equilibrium}
Let $y=x_{\mathrm{des}}-x_{\mathrm{base}}$ and let $H^\dagger$ be the
Moore--Penrose pseudoinverse. The minimum-norm goal is
$v^\dagger=H^\dagger y$, and every minimizing goal has the form
$v=H^\dagger y+(I-H^\dagger H)\zeta$, $\zeta\in\mathbb R^p$. The
unique closest equilibrium and structural distance are
\begin{align}
x^{\mathrm{closest}}&=x_{\mathrm{base}}+HH^\dagger y,\nonumber\\
d_{\mathrm{str}}&=\|(I-HH^\dagger)y\|_2.
\label{eq:closest-and-structural-distance}
\end{align}
For every $v$, orthogonality gives
\begin{equation}
\|x^*(v)-x_{\mathrm{des}}\|_2^2
=\|x^*(v)-x^{\mathrm{closest}}\|_2^2+d_{\mathrm{str}}^2.
\label{eq:structural-within-decomposition}
\end{equation}
\end{proposition}

Here $HH^\dagger$---not the nonidempotent $P_{\mathrm{EQ}}$---is the
relevant orthogonal projector. The first term in
\eqref{eq:structural-within-decomposition} is a within-set mismatch,
not a measure of noncooperative inefficiency; welfare requires a
same-state Nash--social comparison.

\subsection{Social Optimum of the One-Step Game}
\label{subsec:one-step-social-optimum}

Fix the network state and write
\[
    z:=Ax+\Theta x(0).
\]
The total one-step social cost is the sum of the individual player
costs:
\begin{equation}
    \mathcal{J}(u;z)
    :=
    \sum_{m=1}^{p}
    \left[
        \left\|z+Bu-g_m\right\|_2^2
        +
        \gamma_m u_m^2
    \right].
    \label{eq:total-social-cost}
\end{equation}
Let $g_\Sigma:=\sum_mg_m$, $q:=B^\top g_\Sigma$, and
$S_{\mathrm{SO}}:=pG+\Gamma\succ0$. Then
\begin{equation}
    \mathcal{J}(u;z)
    =
    u^{\top}S_{\mathrm{SO}}u
    +2u^{\top}(pB^{\top}z-q)+\operatorname{const}(z),
    \label{eq:expanded-social-cost}
\end{equation}
Thus
\begin{equation}
u^{\mathrm{SO}}=S_{\mathrm{SO}}^{-1}(q-pB^\top z),
\qquad
u^{\mathrm{NE}}=S^{-1}(v-B^\top z).
\label{eq:Nash-and-social-actions}
\end{equation}

\begin{proposition}[Nash equilibrium versus social optimum]
\label{prop:Nash-versus-social-optimum}
For a fixed free-response vector $z$, the potential function and
the total social cost have Hessians $2(G+\Gamma)$ and
$2(pG+\Gamma)$, respectively. With
$e_{\mathrm{comp}}:=u^{\mathrm{NE}}-u^{\mathrm{SO}}$, the additive
welfare loss is
\begin{equation}
\Delta_{\mathrm{comp}}(z)
:=\mathcal J(u^{\mathrm{NE}};z)-\mathcal J(u^{\mathrm{SO}};z)
=e_{\mathrm{comp}}^\top S_{\mathrm{SO}}e_{\mathrm{comp}}\geq0.
    \label{eq:additive-price-of-competition}
\end{equation}
Moreover,
\begin{align}
 e_{\mathrm{comp}}(z)
 &=S^{-1}v-S_{\mathrm{SO}}^{-1}q
 +(p-1)S^{-1}\Gamma S_{\mathrm{SO}}^{-1}B^\top z,
 \label{eq:competitive-displacement-closed-form}\\
 \Delta_{\mathrm{comp}}(z)
 &=\|pv-q-(p-1)\Gamma u^{\mathrm{NE}}(z)\|_{S_{\mathrm{SO}}^{-1}}^2.
 \label{eq:competition-loss-residual-form}
\end{align}
It follows that $\Delta_{\mathrm{comp}}$ is a convex quadratic in
$z$. For $p\geq2$ and $B\ne0$, its zero set is a proper affine set
(possibly empty) with direction $\ker(B^\top)$; for $p=1$, the loss
vanishes identically.
\end{proposition}

\begin{proof}
Complete the square in \eqref{eq:expanded-social-cost}, then use the
two first-order conditions and
$pS_{\mathrm{SO}}^{-1}-S^{-1}
=(p-1)S^{-1}\Gamma S_{\mathrm{SO}}^{-1}$.
\end{proof}

When
\[
    \mathcal{J}\bigl(u^{\mathrm{SO}}(z);z\bigr)>0,
\]
the corresponding multiplicative price of competition is
\begin{equation}
    \Pi_{\mathrm{comp}}(z)
    :=
    \frac{
        \mathcal{J}\bigl(u^{\mathrm{NE}}(z);z\bigr)
    }{
        \mathcal{J}\bigl(u^{\mathrm{SO}}(z);z\bigr)
    }
    =
    1+
    \frac{
        \Delta_{\mathrm{comp}}(z)
    }{
        \mathcal{J}\bigl(u^{\mathrm{SO}}(z);z\bigr)
    }
    \geq1.
    \label{eq:multiplicative-PoC}
\end{equation}
The additive measure is used whenever the socially optimal cost
vanishes or when a ratio would be numerically ill-conditioned. These
are instance- and state-specific measures, rather than a worst-case
price-of-anarchy bound over a class of games.

\subsection{Price of Competition at the Closed-Loop Equilibrium}
\label{subsec:equilibrium-price-of-competition}

At the Nash equilibrium, let
$z^*_{\mathrm{NE}}:=Ax^*(v)+\Theta x(0)$ and define
\begin{align}
    \Delta_{\mathrm{comp}}^{*}(v)
    &:=
    \Delta_{\mathrm{comp}}
    \bigl(z^{*}_{\mathrm{NE}}\bigr),
    \\
    \Pi_{\mathrm{comp}}^{*}(v)
    &:=
    \Pi_{\mathrm{comp}}
    \bigl(z^{*}_{\mathrm{NE}}\bigr),
\end{align}
when the ratio is defined. This same-state comparison guarantees
$\Delta_{\mathrm{comp}}^{*}(v)\geq0$.
\subsection{Centralized Social-Feedback Benchmark}
\label{subsec:social-feedback-benchmark}

Repeated application of the social action gives
\begin{align}
    u^{\mathrm{SO}}(k)
    &=S_{\mathrm{SO}}^{-1}
      [q-pB^\top(Ax(k)+\Theta x(0))],
    \nonumber\\
    P_{\mathrm{SO}}&:=pBS_{\mathrm{SO}}^{-1}B^\top,
    \nonumber\\
    F_{\mathrm{SO}}
    &:=
    (I-P_{\mathrm{SO}})A
    =(I+pB\Gamma^{-1}B^\top)^{-1}A,
    \\
    c_{\mathrm{SO}}
    &:=
    (I-P_{\mathrm{SO}})\Theta x(0)
    +
    BS_{\mathrm{SO}}^{-1}q.
\end{align}
Thus, when $I-F_{\mathrm{SO}}$ is nonsingular,
$x_{\mathrm{SO}}^*=(I-F_{\mathrm{SO}})^{-1}c_{\mathrm{SO}}$. The
outcome displacement
$d_{\mathrm{comp}}^{\mathrm{ss}}:=\|x^*(v)-x_{\mathrm{SO}}^*\|_2$
is distinct from the same-state welfare measures above. Comparing
$F_{\mathrm{EQ}}=(I+C)^{-1}A$ with
$F_{\mathrm{SO}}=(I+pC)^{-1}A$ shows that the planner internalizes the
same influence Gramian at $p$ times the strength.

\section{Numerical Experiments}
\label{sec:numerical-experiments}

The following deterministic examples illustrate the stability,
geometry, and welfare results. All reported spectral radii were
computed directly from the corresponding matrices. Code reproducing
the numerical analyses is publicly available on
\href{https://github.com/GabrielGent/BR_OSAOC.git}{GitHub}.

\subsection{Primitive Marginal Network}
\label{subsec:numerical-primitive-network}

Consider the primitive row-stochastic matrix
\begin{equation}
    A_{\mathrm{pr}}
    =
    \begin{bmatrix}
        0.5 & 0.5 & 0   \\
        0.2 & 0.5 & 0.3 \\
        0.1 & 0.3 & 0.6
    \end{bmatrix}.
    \label{eq:primitive-numerical-A}
\end{equation}
Its spectrum is
$\{1,(3+\sqrt3)/10,(3-\sqrt3)/10\}$. Since this matrix is irreducible
and marginal, Lemma~\ref{lem:marginal-fj} forces $\Theta=0$. We use
\begin{equation}
    B_{\mathrm{pr}}
    =
    \begin{bmatrix}
        1   & 0.2 \\
        0.5 & 0.8 \\
        0.2 & 1
    \end{bmatrix}.
    \label{eq:primitive-numerical-B}
\end{equation}
With $\Gamma(t)=tI_2$, the normalized left Perron vector is
$\pi=(11,20,15)^\top/46$ and the damping coefficient is
$\pi^\top B_{\mathrm{pr}}B_{\mathrm{pr}}^\top\mathbf1
=268/115=2.330435\ldots>0$.
For $t=0.1,1,10,100,1000$, the corresponding spectral radii are
$0.140270,0.299121,0.810623,0.977221,0.997675$. Thus the closed loop
is stable and approaches the unit circle from below (Fig.\ref{fig:EQ-stability-sweeps}, graph on the left); numerically,
$t[1-\rho(F_{\mathrm{EQ}}(t))]\to2.330435$, as predicted by the
first-order analysis.
At the other end of the sweep, Lemma~\ref{lem:resolvent} gives
$\rho((I-\Pi_{B_{\mathrm{pr}}})A_{\mathrm{pr}})=0.134049$ as
$t\downarrow0$.

\subsection{An Imprimitive Period-Two Network and the Failure
of Unsigned Alignment}
\label{subsec:numerical-periodic-network}

We next consider the periodic irreducible stochastic matrix
\begin{equation}
    A_{\mathrm{per}}
    =
    \begin{bmatrix}
        0 & 0 & 1 \\
        0 & 0 & 1 \\
        0.9 & 0.1 & 0
    \end{bmatrix},
    \qquad
    \sigma(A_{\mathrm{per}})
    =
    \{1,-1,0\}.
    \label{eq:periodic-example-A}
\end{equation}
There is one player, with
\begin{equation}
    B_{\mathrm{per}}
    =
    b
    =
    \begin{bmatrix}
        0\\10\\5
    \end{bmatrix},
    \qquad
    \Gamma=[\gamma].
    \label{eq:periodic-example-B}
\end{equation}
For $w_1=(1,1,1)^\top$ and $w_{-1}=(1,1,-1)^\top$, both unsigned
alignments are nonzero: $b^\top w_1=15$ and $b^\top w_{-1}=5$.
However, the normalized left eigenvector
$y_{-1}=(0.45,0.05,-0.5)^\top$ gives
$\alpha_{-1,1}=(y_{-1}^\top b)(b^\top w_{-1})=-10<0$, using the
columnwise convention of \eqref{eq:columnwise-damping} (equivalently,
$\bar\gamma=1$ in Theorem~\ref{thm:EQ-large-penalty-stability}). The
characteristic polynomial is
\begin{equation}
\det(\lambda I-F_{\mathrm{EQ}}(\gamma))
=\frac{\lambda[(\gamma+125)\lambda^2+55\lambda-(\gamma+90)]}
{\gamma+125}.
\label{eq:periodic-characteristic-polynomial}
\end{equation}
Its normalized quadratic factor satisfies
$q_\gamma(-1)=-20/(\gamma+125)<0$, so one
eigenvalue is below $-1$ for every finite $\gamma>0$. For
$\gamma=0.1,10,1000,10000$, the spectral radii are respectively
$1.096492,1.088145,1.009069,1.000990$: they approach one from above (Fig.\ref{fig:EQ-stability-sweeps}, graph on the right),
with $\rho(F_{\mathrm{EQ}}(\gamma))=1+10/\gamma+O(\gamma^{-2})$.
This single-player example tests stability, not strategic inefficiency;
indeed its welfare loss is identically zero.

\begin{figure*}[t]
    \centering
    \includegraphics[width=0.80\textwidth]
    {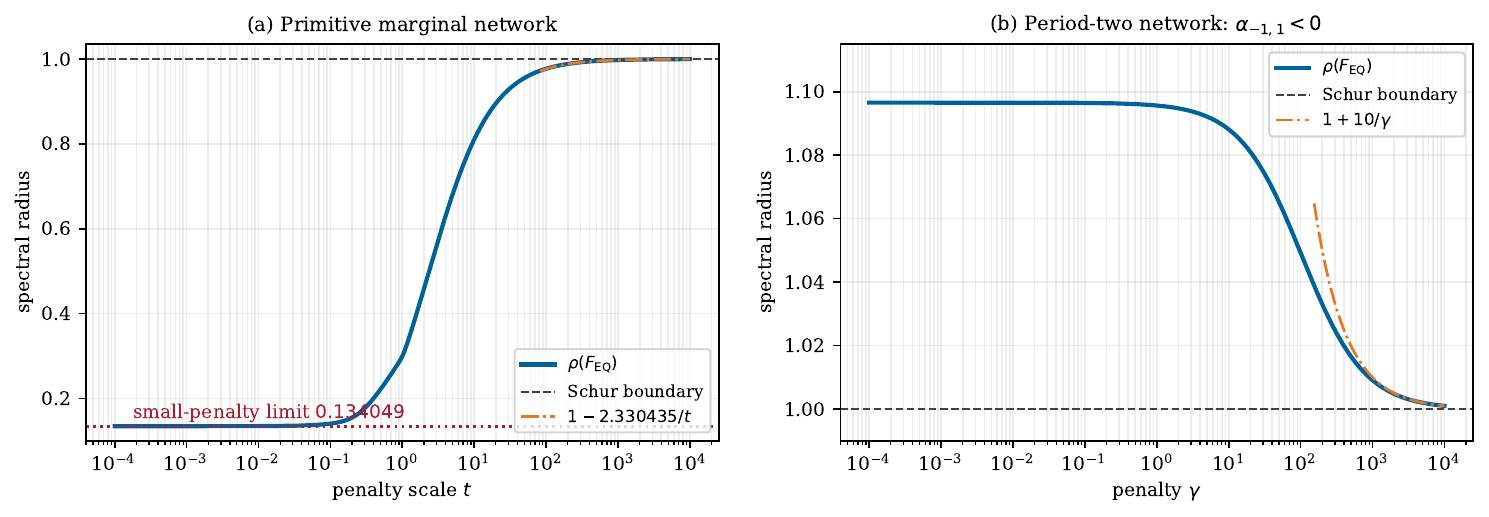}
    \caption{Penalty sweeps for exact Nash-equilibrium feedback.
    Left: primitive network, for which the spectral radius approaches
    one from below. Right: periodic network with a negatively damped
    $-1$ mode, for which the spectral radius approaches one from
    above. The line marks the Schur boundary; annotations show the
    small-penalty limit and large-penalty first-order asymptotes.}
    \label{fig:EQ-stability-sweeps}
\end{figure*}

\subsection{Frozen-State Convergence versus Closed-Loop Stability}
\label{subsec:numerical-augmented-stability}

We now use $A=A_{\mathrm{pr}}$ from
\eqref{eq:primitive-numerical-A}, three players, and identical
influence directions:
\begin{equation}
    B_{\mathrm{it}}
    =
    \begin{bmatrix}
        1&1&1\\
        1&1&1\\
        1&1&1
    \end{bmatrix},
    \qquad
    \Gamma=\gamma I_3.
    \label{eq:iterative-example-B}
\end{equation}
The influence directions overlap completely and
$B_{\mathrm{it}}^\top B_{\mathrm{it}}=3\mathbf1\mathbf1^\top$.

For the frozen-state PBR iteration,
\begin{equation}
    \rho(T_{\mathrm{PBR}})
    =
    \frac{6}{3+\gamma}.
    \label{eq:iterative-PBR-exact-radius}
\end{equation}
The spectral radii in this experiment are independent of $v$, which
only changes the affine term.

\begin{proposition}[Exact PBR thresholds in this experiment]
\label{prop:iterative-thresholds}
For \eqref{eq:primitive-numerical-A} and
\eqref{eq:iterative-example-B},
\begin{align}
\det(\lambda I-\mathcal A_{\mathrm{PBR}})
&=\frac{(50\lambda^2-30\lambda+3)
((\gamma+3)\lambda-3)^2}{50(\gamma+3)^3}\nonumber\\[-1mm]
&\quad{}\times
((\gamma+3)\lambda^2+(12-\gamma)\lambda-6).
\label{eq:pbr-augmented-characteristic}
\end{align}
Consequently frozen PBR converges if and only if $\gamma>3$, whereas
the augmented PBR loop is Schur if and only if $\gamma>15/2$.
\end{proposition}

\begin{proof}
The frozen condition follows from
\eqref{eq:iterative-PBR-exact-radius}. In
\eqref{eq:pbr-augmented-characteristic}, the first two factors have
roots in the unit disk for every $\gamma>0$. Applying the quadratic
Jury conditions to the last factor yields $\gamma>3$, a vacuous
positive inequality, and the binding condition $2\gamma-15>0$.
\end{proof}

Thus the exact separation gap is $\gamma\in(3,15/2]$, whereas SBR
converges for every $\gamma>0$. At $\gamma=6$,
$\rho(T_{\mathrm{PBR}})=0.666667$ but
$\rho(\mathcal A_{\mathrm{PBR}})=1.215250$, so frozen-state PBR
convergence coexists with an unstable one-sweep closed loop. For the
same parameters, $\rho(T_{\mathrm{SBR}})=0.192450$ and
$\rho(\mathcal A_{\mathrm{SBR}})=0.521584$; in fact
$\rho(T_{\mathrm{SBR}})=3^{-3/2}$. Figure~\ref{fig:iterative-stability-sweeps}
shows the complete sweeps.

\begin{figure*}[t]
    \centering
    \includegraphics[width=0.80\textwidth]
    {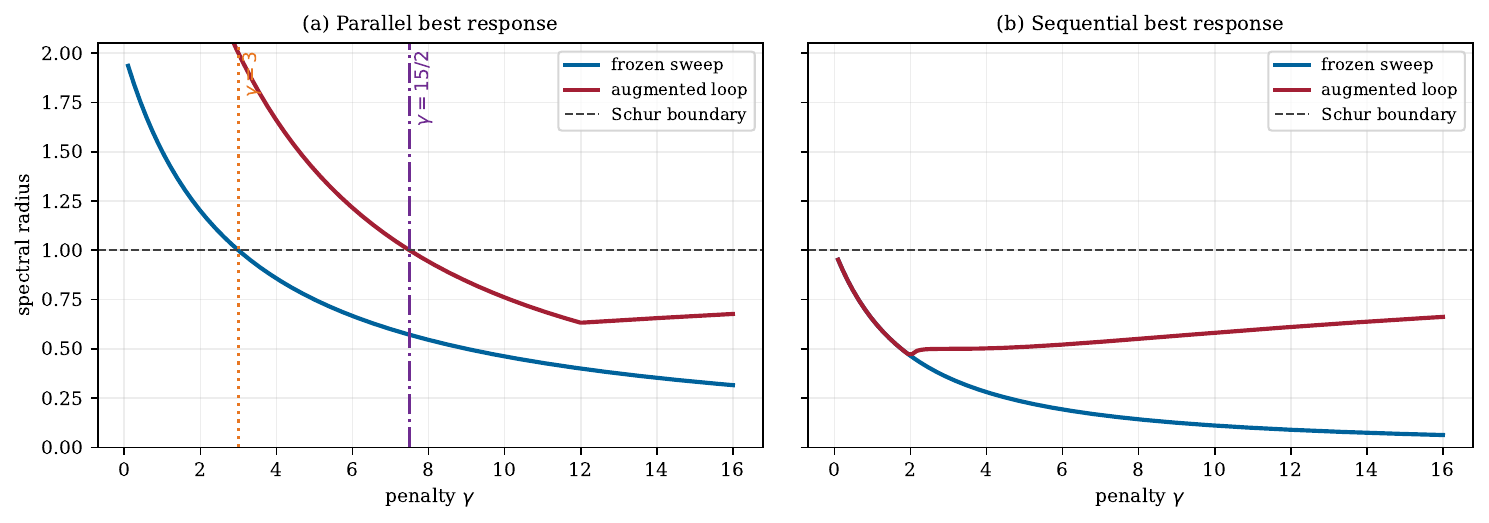}
    \caption{Frozen-state and augmented spectral radii for PBR and
    SBR as functions of $\gamma$. Vertical markers identify the exact
    frozen PBR and augmented PBR thresholds, $3$ and $15/2$.}
    \label{fig:iterative-stability-sweeps}
\end{figure*}

\subsection{Structural and Competitive Inefficiency}
\label{subsec:numerical-competition-metrics}

Finally, we return to the primitive matrices
$A_{\mathrm{pr}}$ and $B_{\mathrm{pr}}$ in
\eqref{eq:primitive-numerical-A} and
\eqref{eq:primitive-numerical-B}. We use
\[
    \Gamma=I_2,
    \qquad
    \Theta=0,
\]
and choose
\begin{equation}
    g_1=
    \begin{bmatrix}
        1\\0\\0
    \end{bmatrix},
    \qquad
    g_2=
    \begin{bmatrix}
        0\\0\\1
    \end{bmatrix}.
    \label{eq:competition-example-goals}
\end{equation}
Here $v=(1,1)^\top$ and $q=(1.2,1.2)^\top$. The Nash and centralized
closed loops are stable, with spectral radii $0.299121$ and $0.221680$,
and their steady states are separated by
$d_{\mathrm{comp}}^{\mathrm{ss}}=0.371726$. At the Nash equilibrium
state, the Nash and social costs are $1.575607$ and $1.348770$;
therefore $\Delta_{\mathrm{comp}}^*=0.226837$ and
$\Pi_{\mathrm{comp}}^*=1.168181$. Thus noncooperative play raises the
same-state one-step social cost by approximately $16.8\%$.

For $x_{\mathrm{des}}=(g_1+g_2)/2$, the structural and within-set
distances are $0.401416$ and $0.354051$, while the total target error
is $0.535244$. The identity
$(0.535244)^2=(0.401416)^2+(0.354051)^2$ holds up to rounding,
confirming \eqref{eq:structural-within-decomposition}.

\subsection{Goal Conflicts Restricted to Selected Nodes}
\label{subsec:numerical-constrained-goals}

We reuse $A_{\mathrm{pr}}$, $B_{\mathrm{pr}}$, and $\Gamma=I_2$ from the
primitive-network experiment.  For two players, let
$U_\perp=2^{-1/2}[1\ {-1}]^\top$, choose the common goal
$\bar g=0.5\mathbf{1}$, and introduce conflict at node $i$ through
\begin{equation}
\begin{aligned}
    g_1(\eta)&=\bar g+\frac{\eta}{\sqrt{2}}e_i^{(n)},
    &g_2(\eta)&=\bar g-\frac{\eta}{\sqrt{2}}e_i^{(n)},\\
    |\eta|&\leq\varepsilon:=\frac{1}{\sqrt{2}}.
\end{aligned}
    \label{eq:numerical-supported-goals}
\end{equation}
The bound keeps all goal components in $[0,1]$. At $\eta=0$,
$v_0=(0.85,1)^\top$, $x^*(v_0)=z_0=0.5\mathbf1$, and both Nash and
social actions vanish. With
$D_i:=\operatorname{diag}(B_{i1},B_{i2})$, define
\[
\begin{aligned}
    \kappa_i&:=\|HD_iU_\perp\|_2,\\
    \beta_i&:=(S^{-1}D_iU_\perp)^\top
             S_{\mathrm{SO}}(S^{-1}D_iU_\perp).
\end{aligned}
\]
Since $g_1(\eta)+g_2(\eta)=\mathbf{1}$, the social optimum at $z_0$
remains zero, whereas
\begin{equation}
\begin{aligned}
    \|x^*(\eta)-x^*(v_0)\|_2&=\kappa_i|\eta|,\\
    \Delta_{\mathrm{comp}}^{(i)}(z_0;\eta)&=\beta_i\eta^2.
\end{aligned}
    \label{eq:numerical-conflict-verification}
\end{equation}

\begin{table}[t]
    \centering
    \caption{Perron weight, conflict amplification, and welfare loss
    evaluated at $z_0$.}
    \label{tab:constrained-goal-nodes}
    \setlength{\tabcolsep}{3.5pt}
    \begin{tabular}{c|cccc}
        \hline
        Node $i$ & $\pi_i$ & $\kappa_i$ &
        $\varepsilon\kappa_i$ &
        $\Delta_{\mathrm{comp}}^{(i)}(z_0;\pm\varepsilon)$\\
        \hline
        1 & 0.239130 & 0.489183 & 0.345905 & 0.201380\\
        2 & 0.434783 & 0.449150 & 0.317597 & 0.178809\\
        3 & 0.326087 & 0.480138 & 0.339509 & 0.182720\\
        \hline
    \end{tabular}
\end{table}

Node 2 has the largest Perron weight, but node 1 has the largest
equilibrium gain and endpoint welfare loss. Each $HD_iU_\perp$ has
rank at most one by Proposition~\ref{prop:support-restricted-goals},
and exactly one here because $D_iU_\perp\ne0$ and
$H=K^{-1}BS^{-1}$ is injective for this full-column-rank $B$.
Thus network centrality need not rank conflict amplification, whereas
the control-aware gain predicts the exact worst-case displacement.

\section{Conclusion}
\label{sec:conclusion}

Competitive one-step-ahead control of Friedkin--Johnsen networks is an
exact potential game even under overlapping influence and conflicting
goals. Its common Nash fixed point must be distinguished from protocol
implementation: frozen-state best-response convergence does not imply
one-sweep closed-loop stability, nor conversely. Sequential frozen
sweeps always converge, parallel sweeps always converge for two
players, and the exact three-player example identifies a nonempty gap
between frozen and augmented PBR thresholds. For exact-equilibrium
feedback, the resolvent form reveals the feedback geometry, while
signed left--right modal damping gives a first-order sharp
large-penalty stability test.

The equilibrium family is governed by the influence rank, while
support, amplitude, and sparsity restrictions on goal conflict determine
its constrained geometry. Projection separates structural target error
from within-set mismatch, and the same-state cost gap isolates welfare
loss. The examples confirm these distinctions and show that network
centrality alone need not predict conflict amplification. Extensions to
repeated peripheral eigenvalues and multi-step actions remain open.
Other natural directions include the zero first-order damping case,
constrained player actions, and trajectory-level rather than
same-state welfare comparisons.

\appendices
\section{Proofs of the Large-Penalty Results}
\label{app:large-penalty-proofs}

\begin{proof}[Proof of Theorem~\ref{thm:EQ-large-penalty-stability}]
Let $G=B^\top B$ and $\varepsilon=t^{-1}$. A Neumann expansion gives
\[
(G+t\bar\Gamma)^{-1}
=\varepsilon\bar\Gamma^{-1}+O(\varepsilon^2),
\]
and hence, with $C=B\bar\Gamma^{-1}B^\top$,
\begin{equation}
F_{\mathrm{EQ}}(t)=A-\varepsilon CA+O(\varepsilon^2).
\label{eq:appendix-F-expansion}
\end{equation}
If $\rho(A)<1$, continuity of the eigenvalues proves the first claim.
Now let $\rho(A)=1$. Strictly stable eigenvalues remain inside the
unit disk for small $\varepsilon$. For a simple peripheral eigenvalue,
first-order perturbation theory \cite[Thm.~2.3, p.~183]{StewartSun90}
and $Aw_\mu=\mu w_\mu$ yield
\[
\lambda_\mu(\varepsilon)
=\mu(1-\varepsilon y_\mu^*Cw_\mu)+O(\varepsilon^2),
\]
so, because $|\mu|=1$,
\[
|\lambda_\mu(\varepsilon)|^2
=1-2\varepsilon\operatorname{Re}(y_\mu^*Cw_\mu)
+O(\varepsilon^2).
\]
Condition~\eqref{eq:spectral-damping-condition} therefore moves every
peripheral eigenvalue strictly inside the unit disk for sufficiently
small positive $\varepsilon$. Finiteness of the peripheral spectrum
provides a common bound. If instead $\alpha_\mu<0$, the same expansion
gives
$|\lambda_\mu(\varepsilon)|
=1+\varepsilon|\alpha_\mu|+O(\varepsilon^2)>1$ for all sufficiently
small positive $\varepsilon$, proving
\eqref{eq:large-penalty-instability-converse}.
\end{proof}

\begin{proof}[Proof of Corollary~\ref{cor:independent-penalty-stability}]
Remove zero columns of $B$, set
$\delta=(\min_{m\in\mathcal I}\gamma_m)^{-1}$, and write
$\Gamma^{-1}=\delta R$, where
$R=\operatorname{diag}(r_m)_{m\in\mathcal I}$,
$0<r_m\leq1$, and $\max_{m\in\mathcal I}r_m=1$. The closure
\[
\mathcal R=\{R:0\leq r_m\leq1,\ 
\max_{m\in\mathcal I}r_m=1\}
\]
is compact. The inverse expansion below has a remainder uniform on
$\mathcal R$:
\[
F_{\mathrm{EQ}}(\Gamma)
=A-\delta BRB^\top A+O(\delta^2).
\]
For every peripheral mode,
\[
\operatorname{Re}(y_\mu^*BRB^\top w_\mu)
=\sum_{m\in\mathcal I} r_m\alpha_{\mu m}>0.
\]
Since at least one $r_m$ equals one and all $\alpha_{\mu m}$ are
positive, this expression is uniformly bounded away from zero. The
perturbation argument used above is therefore uniform in $R$ and
applies whenever $\delta$ is sufficiently small.
\end{proof}

	\bibliographystyle{IEEEtran}
	\bibliography{refs}
\end{document}